\documentclass[11pt]{article}
\usepackage[numbers]{natbib}
\usepackage[affil-it]{authblk}
\usepackage{fullpage, todonotes, scalerel, mathtools}
\usepackage[boxruled,noend,nofillcomment]{algorithm2e}
\usepackage{tikz}
\usetikzlibrary{patterns}
\usepackage{pgfplots}
\usepackage{color-edits}
\addauthor{jm}{magenta}
\addauthor{avrim}{green}
\addauthor{adam}{red}
\addauthor{ankit}{blue}

\usepackage{amsthm}
\usepackage{thm-restate}
\usepackage{cleveref}

\newcommand{\eps}{\epsilon}
\newcommand{\re}{\mathbb{R}}
\newcommand{\en}{\mathbb{N}}

\newcommand{\va}{{\vec{a}}}
\newcommand{\Lap}{{\mathsf{Lap}}}

\newtheorem{prop}{Proposition}
\usepackage{amsmath, amsfonts}
\usepackage{xspace}
\newtheorem{obs}{Observation}

\newtheorem{lemma}{Lemma}
\newtheorem{definition}{Definition}
\newtheorem{theorem}{Theorem}
\newtheorem{corollary}{Corollary}

\newcommand{\sig}[1]{\ensuremath{m_{#1}(a_1, \ldots, a_{#1-1})}\xspace}
\newcommand{\ssig}[1]{\ensuremath{m_{#1}}\xspace}
\newcommand{\util}[1]{\ensuremath{u_{#1}(a_1, \ldots, a_{n})}\xspace}

\newcommand{\ract}[2]{\ensuremath{a_{#1,#2}\xspace}}

\newcommand{\mech}{\ensuremath{\mathcal{M}}\xspace}
\renewcommand{\exp}[1]{\ensuremath{\mathbb{E}\left[ #1 \right]}}
\newcommand{\prob}{\ensuremath{\mathbb{P}}\xspace}
\newcommand{\mult}{\ensuremath{\alpha}\xspace}
\newcommand{\add}{\ensuremath{\beta}\xspace}
\newcommand{\pfail}{\ensuremath{\gamma}\xspace}

\newcommand{\used}[2]{\ensuremath{x_{#1, #2}}\xspace}

\newcommand{\w}[1]{\ensuremath{x_{#1}}\xspace}
\newcommand{\sused}[2]{\ensuremath{y_{#1, #2}}\xspace}
\newcommand{\Sused}[1]{\ensuremath{y_{#1}}\xspace}

\newcommand{\nnull}{\ensuremath{\mech_\emptyset}\xspace}
\newcommand{\full}{\ensuremath{\mech_{Full}}\xspace}
\newcommand{\shallow}[2]{\ensuremath{(#1,#2)}-shallow\xspace}
\newcommand{\vol}{\ensuremath{w}\xspace}
\newcommand{\length}{\ensuremath{l}\xspace}
\newcommand{\val}[2]{\ensuremath{v_{#2}(#1)}\xspace}

\newcommand{\dval}[2]{\ensuremath{V_{#2}(#1)}\xspace}
\newcommand{\dvalo}[1]{\ensuremath{V_{#1}}\xspace}
\newcommand{\comp}{\ensuremath{CR}\xspace}
\newcommand{\ud}{\textsc{Undom}\xspace}
\newcommand{\greedy}{\textsc{Greedy}\xspace}

\newcommand{\counters}{\textsc{counters}\xspace}

\newcommand{\Z}{\ensuremath{\mathbb{Z}}\xspace}
\newcommand{\R}{\ensuremath{\mathbb{R}}\xspace}

\newcommand{\norminfty}[1]{\lvert #1 \rvert_{\infty}}

\renewcommand{\hat}[1]{\widehat{#1}}
\begin{document}

\author[1]{Avrim Blum \thanks{Blum, Morgenstern, and Sharma were
    partially supported by NSF grants CCF-1116892 and
    CCF-1101215. Morgenstern was partially supported by an NSF GRFP
    award and a Simons Award for Graduate Students in Theoretical
    Computer Science.}} \author[1]{Jamie Morgenstern} \author[1]{Ankit
  Sharma} \author[2]{Adam Smith}

\affil[1]{Computer Science Department\\
Carnegie Mellon University\\
}
\affil[2]{Computer Science Department\\
Penn State University}

\title{Privacy-Preserving Public Information for Sequential Games}
\maketitle

\begin{abstract}
In settings with incomplete information, players can find it difficult to coordinate to find states with good social welfare. For instance, one of the main reasons behind the recent financial crisis was found to be the lack of market transparency, which made it difficult for financial firms to accurately measure the risks and returns of their investments. Although regulators may have access to firms' investment decisions, directly reporting all firms' actions raises confidentiality concerns for both individuals and institutions. The natural question, therefore, is whether it is possible for the regulatory agencies to publish some information that, on one hand, helps the financial firms understand the risks of their investments better, and, at the same time, preserves the privacy of their investment decisions. More generally, when can the publication of privacy-preserving information about the state of the game improve overall outcomes such as social welfare?

In this paper, we explore this question in a sequential resource-sharing game where the value gained by a player on choosing a resource depends on the number of other players who have chosen that resource in the past.  Without any knowledge of the actions of the past players, the social welfare attained in this game can be arbitrarily bad. We show, however, that it is possible for the players to achieve good social welfare with the help of \emph{privacy-preserving, publicly-announced information}. We model the behavior of players in this imperfect information setting in two ways -- greedy and undominated strategic behaviours, and we prove guarantees on social welfare that certain kinds of privacy-preserving information can help attain. To achieve the social welfare guarantees, we design a counter with improved privacy guarantees under continual observation. In addition to the resource-sharing game, we study the main question for other games including sequential versions of the cut, machine-scheduling and cost-sharing games, and games where the value attained by a player on a particular action is not only a function of the actions of the past players but also of the actions of the future players.
\end{abstract}

\thispagestyle{empty}
\newpage
\clearpage
\setcounter{page}{1}

\section{Introduction}\label{section:intro}
Multi-agent settings that are non-transparent (where players cannot
see the current state of the system) have the potential to lead to
disastrous outcomes.  For example, in examining causes of the recent
financial crisis and subsequent recession, the
\citet[p. 352]{FCIC2011a} concluded that ``The OTC derivatives
market's lack of transparency and of effective price discovery
exacerbated the collateral disputes of AIG and Goldman Sachs and
similar disputes between other derivatives counterparties.''  Even
though regulators have access to detailed confidential information
about financial institutions and (indirectly) individuals, current
statistics and indices are based only on public data, since
disclosures based on confidential information are restricted. However,
forecasts based on confidential data can be much more
accurate\footnote{For example, \citet{OetBiancoGramlichOng2012a}
  compared an index based on both public and confidential data with an
  analogous index based only on publicly available data. The former
  index would have been a significantly more accurate predictor of
  financial stress during the recent financial crisis (see
  \citet[Figure 4]{OetBiancoGramlichOng2011a}). See \citet{FloodKOS13}
  for further discussion.}, prompting regulators to ask whether
aggregate statistics can be economically useful while also providing
rigorous privacy guarantees \cite{FloodKOS13}.

In this work, we show that such {\em privacy-preserving public
  information}, in an interesting class of sequential decision-making
games, can achieve (nearly) the best of both worlds.  In particular,
the goal is to produce information about actions taken by previous
agents that can be posted publicly, preserves all agents'
(differential) privacy, and can significantly improve worst-case
social-welfare. While our models do not directly speak to the highly
complex issues involved in real-world financial decision-making, they
do indicate that in settings involving contention for resources and
first-mover advantages, privacy-preserving public information can be a
significant help in improving  social welfare. In the following
sections, we describe the game setting and the information model.

\subsection{Game Model}\label{sec:game-model}
Consider a setting in which there are $m$ resources and $n$ players. The players arrive online, in an \emph{adversarial} order, one at a time\footnote{For ease of exposition, we rename players such   that player $i$ is the $i$th to arrive.}.  Each player $i$ has some set $A_i$ of resources she is interested in and that is known only to herself. An action $a_i$ of player $i$ is of the form $(\ract{i}{1}, \ldots, \ract{i}{m})$, where $\ract{i}{r} \ge 0$ represents the amount that player $i$ invests in resource $r$, and moreover, $\sum_{j\in   [m]}\ract{i}{j} = 1$. For simplicity, we assume that all $a_{i,r}$ are in $\{0,1\}$ i.e, the unit-demand setting (we study the continuous version where $a_{i,r}$'s can be fractional, but still sum to 1, in Appendix~\ref{section:continuous}). Furthermore, we do not make the assumption that players have knowledge of their position in the sequence, that is, a player need not know how many players have acted before her.

Each resource $r$ has some non-increasing function $\dvalo{r}: \Z^+
\rightarrow \R^+$ indicating the \emph{value, or utility, of this
  resource to the $k$th player who chooses it}.  Therefore, the
utility of player $i$ is $u_i(a_i, a_{1, \ldots, i-1}) = \sum_{r}
\ract{i}{r} \dval{\used{i}{r}}{r}$, where $\used{i}{r} =
\sum_{j=1}^{i-1} \ract{j}{r}$ for each $r$. In this {\bf resource
  sharing} setting, the utility for a player of choosing a certain
resource is a function of the resource and (importantly) the number of
players who have invested in the resource before her (and not after
her)\footnote{In Section~\ref{sec:extensions}, we consider a generalization
where the utility to a player of investing in a particular resource is
a function of the total number of players who have chosen that
resource, including those who have invested after her.}.

\paragraph{Illustrative Example}
For each resource, suppose $\dval{k}{r} = \dval{0}{r}/k$,
where $\dval{0}{r}$ is the initial value of resource $r$. The value of
each resource $r$ drops rapidly as a function of the number of players
who have chosen it so far. If each player $i$ has \emph{perfect
  information} about the investment choices made by the players before
her, the optimal action for player $i$ is to greedily select the
action in $A_i$ of highest utility based on the number of players who
have selected each resource so far. As shown in
Section~\ref{section:discrete}, the resulting social welfare of this
behavior is within a factor of 4 of the optimal. In the case where
each player has {\em no information} about other players' behaviors,
some particularly disastrous sequences of actions might reasonably
occur, leading to very low social welfare. For example, if each player
$i$ has access to a public resource $r$ where $\dval{0}{r} = 1$ and a
private resource $r_i$ where $\dval{0}{r_i} = 1 - \epsilon$, each
might reasonably choose greedily according to $\dval{0}{\cdot}$,
selecting the resource of highest initial value (in this case, $r$).
This would give social welfare of $\ln(n)$, whereas the optimal
assignment would give $n (1- \epsilon)$. Without information about the
game state, therefore, the players may achieve only a 
$O\left(\frac{\ln(n)}{n}\right)$ fraction of the possible welfare.

\subsection{Information Model}

In resource sharing games, players' decisions about their actions will
be best when they know how many players have chosen each resource when
they arrive. The mechanisms we consider, therefore, will publicly
announce some estimate of these counts. 
We consider the trade-off between the privacy lost by
publishing these estimates and the accuracy of the counters in terms
of social welfare. We consider three categories of counters for
publicly posting the estimate of resource usage: perfect, private and
empty counters.

{\bf Perfect Counters}: At all points, the counters display the exact
usage of each resource.

{\bf Privacy-preserving public counters}: At all points, the counters
display an approximate usage of the resources while maintaining
privacy for each player. We define the privacy guarantee in
Section~\ref{sec:privacy}.

{\bf Empty Counters:} At all points, every counter displays the value
0.

\subsection{Players' Behavior}
\label{sec:behavior}
Each player is a utility-maximizing agent and will choose the resource
that, given their beliefs about actions taken by previous players and
the publicly displayed counters, gives them maximum value. We analyze
the game play under two classes of strategies -- greedy and
undominated strategies.
\begin{enumerate}
\item {\bf Greedy strategy:} Under the greedy strategy, a player has
  no outside belief about the actions of previous players and chooses
  the resource that maximizes her utility given the \emph{currently
    displayed (or announced) values of the counters}. Greedy is a
  natural choice of strategy to consider since it is the
  utility-maximizing strategy when the usage counts posted are
  \emph{perfect}.
\item {\bf Undominated Strategy(UD):} Under undominated strategies, we allow players to have any beliefs about the actions of the previous players that are consistent with the displayed value of the counters\footnote{As will become clear in Section~\ref{sec:privacy}, we work with privacy-preserving public counters that display values that can be off from the true usage only in a \emph{bounded} range. Hence with these counters, a player's belief is consistent as long as the belief implies the usage of the resource to be a number that is within the  bounded range of the displayed value. Moreover, with empty counters, any belief about the actions of previous players is a consistent belief.},  and they are allowed to play \emph{any} undominated strategy $a_i$ under this belief. A strategy $a_i$ is \emph{undominated} under a belief, if no other $a'_{i}$ get a strictly higher utility. \footnote{For each counter mechanism we consider, there exists at least one undominated  strategy. For example, with perfect counters, the only consistent belief is that the true value is equal to the displayed value and here the greedy strategy is always undominated; moreover, if the counter mechanism has a nonzero probability of outputting the true value, then again the greedy strategy is undominated under the belief that the displayed value is the true value; if the counter mechanism can display values that are arbitrarily off from the true value, then 
for equal initial values \emph{every} strategy is undominated.}
\end{enumerate}

We analyze the social welfare $SW(a) = \sum_{i} u_i(a)$ generated by an announcement mechanism $\mech$ for a set of strategies $D$ and compare it to the optimal social welfare $OPT$. For a game setting $g$, constituted of a collection of players $[n]$ and their allowable actions $A_i$ (as defined in  Section~\ref{sec:game-model}), $OPT(g)$ is defined as the optimal social welfare that can be achieved by any allocation of resources to the players, where the space of feasible allocations is determined by the setting $g$. In the unit-demand setting, $OPT(g)$  is the maximum weight matching in the bipartite graph $G = (U \cup V, E)$ where $U$ is the set of the $n$ players, $V$ has $n$ vertices for each resource $r$, one of value   $\dval{k}{r}$ for each $k\in [n]$, and there is an edge between player   $i$ and all vertices corresponding to resource $r$ if and only if $r   \in A_{i}$ (Note that the weights are on the vertices in $V$). The object of our study is $\comp_D(g, \mech)$, the \emph{worst case competitive ratio} of the optimal social welfare to the welfare achieved under strategy $D$ and counter mechanism $\mech$. As mentioned earlier, $D$ will either be the greedy ($\greedy$) or the undominated ($\ud$) strategy, and $\mech$ will be either the perfect ($\full$), the privacy-preserving or the empty ($\nnull$) counter.  When $\mech$ uses internal random coins, our results will either be worst-case over all possible throws of the random coins, or will indicate the probability with which the social welfare guarantee holds.
\subsection{Statement of Main Results}\label{section:main}

For sequential resource-sharing games, we prove that for all nonincreasing value curves, the greedy strategy following privacy-preserving counters has a competitive ratio \emph{polylogarithmic} in the number of players (Theorem~\ref{thm:polylog}). This should be contrasted with the competitive ratio of $4$ achieved by greedy w.r.t. perfect counters (Theorem~\ref{thm:greedy4}) and the linear (in the number of players) competitive ratio of greedy with empty counters (as shown in the illustrative example in Section~\ref{sec:game-model}). For the case of undominated strategies, when the marginal values of resources drop slowly, (for example, at a polynomial rate, $\dval{k}{r} = \dval{0}{r}/k^p$ for constant $p>0$), we bound the competitive ratio (w.r.t. privacy-preserving counters) (Theorem~\ref{thm:undominated}). With empty counters, the competitive ratio for undominated strategies is unbounded (Theorem~\ref{thm:noinfo}) for arbitrary curves and is at least quadratic (in the number of players) if the value curve drops slowly (Theorem~\ref{thm:noinfospecial}). We note here that for many of our positive results for privacy preserving counters state the competitive ratio in terms of parameters of the counter vector $\alpha$ and $\beta$ (as detailed in Section~\ref{sec:privacy}) and for a particular implementation of the counter vectors, the values of $\alpha$ and $\beta$ are mentioned in Section~\ref{section:counters}.

%
The key privacy tool we use is the differentially private counter under continual observation~\citep{dwork2010}, which we use to publish estimates of the usage of each resource. We improve upon the existing error guarantees of differentially private counters and design a new differentially private counter in Section~\ref{section:counters}. The new counter provides a tighter additive guarantee at the price of introducing a constant multiplicative error.

In Section~\ref{sec:extensions}, we consider other classes of games -- specifically, we analyze Unrelated Machine Scheduling, Cut, and Cost Sharing games. The work of \citet{Leme:2012} showed these games have improved \emph{sequential} price of anarchy over the \emph{simultaneous} price of anarchy. For these games, we ask the question: if players do not have perfect information to make decisions, but instead have only noisy approximations (due to privacy considerations), does sequentiality still improve the quality of play? We prove that the answer is affirmative in most cases, and
furthermore, for some instances, having \emph{differentially-private   information dissemination improves the competitive ratio over perfect information} (Proposition~\ref{prop:private-better-than-perfect}).


\subsection{Related Work}\label{section:rw}
A great deal of work has been done at the intersection of mechanism
design and privacy; Pai and Roth~\citep{Pairoth13} have an extensive
survey. Our work is similar to much of the previous work in that it
considers maintaining differential privacy to be a constraint. The
focus of our work however is on \emph{how useful information can be
  provided to players in games of imperfect information} to help
achieve a good social objective while respecting the privacy
constraint of the players. The work of \citet{Kearns:2012} is close in
spirit to ours. \citet{Kearns:2012} consider games where players have
incomplete information about other players' types and behaviors. They
construct a privacy-preserving mechanism which collects information
from players, computes an approximate correlated equilibria, and then
advises players to play according to this equilibrium. The
mechanism is approximately incentive compatible for the
players to participate in the mechanism and to follow its
suggestions. Several later papers \citep{congestionrogersR13,
  hsuwalrasian13} privately compute approximate equillibria in
different settings. Our main privacy primitive is the differentially
private counters under continual observation~\citep{dwork2010,
  ChanSS11}, also used in much of the related work on private
equilibrium computation.

Our investigation of cut
games, unrelated machine scheduling, and cost-sharing (Section~\ref{sec:extensions}) is inspired by
work of \citet{Leme:2012}. Their work focuses on the improvement in
social welfare of equilibria in the \emph{sequential} versus the
\emph{simultaneous} versions of certain games. We ask a related
question: when we consider sequential versions of games, and only
\emph{private, approximate information about the state of play} (as opposed to perfect) is
given to players, how much worse can social welfare be?

As mentioned in Section~\ref{sec:behavior}, one class of player behavior for
which we analyze the games is \emph{greedy}. Our analysis of greedy
behavior is in part inspired by the work of~\citet{balcan2009}, who
study best response dynamics with respect to noisy cost functions for
potential games. An important distinction between their setting and ours
is that the noisy estimates we consider are \emph{estimates of state,
  not value}, and may for natural value curves be quite far from
correct in terms of the {\em values} of the actions.

%


%

 
\section{Privacy-preserving public counters}
\label{sec:privacy}
We design announcement mechanisms $\mech_i$ which give approximate
information about actions made by the previous players to player
$i$. Let $\Delta_m$ denote the action space for each player (the $m$-dimensional simplex 
$\Delta_m = \{a \in [0,1]^m \mid \|a\|_1\leq 1\}$).
Mechanism $\mech_i: \left({\Delta_m}\right)^{i-1}\times R \to {\Delta_m}$
depends upon the actions taken before $i$ (specifically, the usage of
each resource by each player), and  on internal random
coins $R$. When player $i$ arrives, $\sig{i}\sim \mech_{i}(a_1,
\ldots, a_{i-1})$ is publicly announced. Player $i$ plays according to
some strategy $d_i : {\Delta_m} \to A_i$, that is $a_i = d_i(m_1, \ldots,
\sig{i})$, a random variable which is a function of this
announcement. When it is clear from context, we denote $\sig{i}$ by
$\ssig{i}$. Formally, the counters used in this paper satisfy the
following notion of privacy.

\begin{definition}
\label{def:privacy}
An announcement mechanism $\mech$ is $(\epsilon,
\delta)$-differentially private under adaptive\footnote{Adaptivity is
  needed in this case because the announcements are arguments to the
  actions of players: when a particular action changes, this modifies
  the distribution over the future announcements, which in turn
  changes the distribution over future selected actions. } continual
observation in the strategies of players if, for each $d$, for each
player $i$, each pair of strategies $d_i , d'_i$, and every
$S\subseteq ({\Delta_m})^n$: \[\prob[(\ssig{1}, \ldots, \ssig{n})\in S]
\leq e^\epsilon\prob[(\ssig{1}, \ldots, \ssig{i}, \ssig{i+1}' \ldots,
\ssig{n}')\in S] + \delta\] where $\ssig{j}\sim \mech_j(a_1, \ldots,
a_{j-1})$ and $\ssig{j}' \sim \mech_j(a_1, \ldots, a_{i-1}, a'_i,
a'_{i+1}, \ldots, a'_{j-1})$, $a_j = d_j(m_1,\ldots, m_j)$, and $a'_i
= d'_i(m_1, \ldots, m_i)$, and for all $j>i$, $a'_{j} = d_{j}(m_1,
\ldots, m_{i-1}, m_i, m'_{i+1}, \ldots, m'_{j})$.
\end{definition}

This definition requires that two worlds which differ in a single player
changing her strategy from $d_i$ to $d'_i$ have statistically close joint
distributions over all players' announcements (and thus their joint
distributions over actions). Note that
the distribution of $j > i$'s announcement can change slightly,
causing $j$'s distribution over actions to change slightly,
necessitating the cascaded $m'_j, a'_j$ for $j > i$ in our
definition. The mechanisms we use maintain approximate use counters
for each resource.  The values of the counters are \emph{publicly
  announced} throughout the game play. We now define the notion of
accuracy used to describe these counters.

\begin{definition}[$(\alpha, \beta, \gamma)$-accurate counter vector]
\label{def:accuracy}
A set of counters $y_{i,r}$ is defined to be $(\mult, \add,
\pfail)$-accurate if with probability at least $1-\gamma$, at all
points of time, the displayed value of every counter $y_{i,r}$ lies in
the range $[\frac{x_{i,r}}{\mult} - \add, \mult x_{i,r} + \add]$ where
$x_{i,r}$ is the \emph{true} count for resource $i$, and is
monotonically increasing in the true count.
\end{definition}

We refer to a set of $(\mult, \add, 0)$-accurate counters as $(\mult,
\add)$-counters for brevity.  It is possible to achieve $\pfail =0$
(which is necessary for undominated strategies, which assumes the
multiplicative and additive bounds on $y$ are worst-case), taking an
appropriate loss in the privacy guarantees for the counter
(Proposition~\ref{prop:zero-failure}). Counters
satisfying Definitions~\ref{def:privacy} and ~\ref{def:accuracy} with
$\mult=1$ and $\add=O(\log^2 n)$ were given in
\citet{dwork2010, ChanSS11}; we give a different implementation in
Section~\ref{section:counters} which gives a tighter bound on $\mult
\add$ by taking $\mult$ to be a small constant larger than 1. Furthermore, the counters in Section~\ref{section:counters} are \emph{monotonic} (i.e., the displayed values can only increase as the game proceeds) and we use monotonicity of the counters in some of our results.

In some settings we require counters we a more specific utility guarantee:

\begin{definition}[$(\alpha, \beta, \gamma)$-accurate underestimator
  counter vector]
  A set of counters $y_{i,r}$ is defined to be $(\mult, \add,
  \pfail)$-accurate if with probability at least $1-\gamma$, at all
  points of time, the displayed value of every counter $y_{i,r}$ lies
  in the range $[\frac{x_{i,r}}{\mult} - \add, x_{i,r}]$ where
  $x_{i,r}$ is the \emph{true} count for resource $i$.
\end{definition}

\noindent The following observation states that a counter vector can
be converted to an undercounter with small loss in accuracy.

\begin{obs}\label{obs:underestimator}
  We can convert a $(\mult, \add)$-counter to an $\left(\mult^2,
    \frac{2\add}{\mult}\right)$-underestimating counter vector.
\end{obs}

\begin{proof}
  We can shift the counter, $\frac{1}{\mult}x - \add \leq y \leq 
  \mult x + \add$ implies $y ' = \frac{y - \add}{\mult}\leq x$ and
  $ \frac{1}{\mult^2}x - \frac{2 \add}{\mult} \leq y'$.
\end{proof}


\section{Resource Sharing}\label{section:discrete}

In this section, we consider resource sharing games -- the utility to
a player is completely determined by the resource she chooses and the
number of players who have chosen that resource before her. This
section considers the case where players' actions are discrete: $a_i
\in \{0,1\}^m$ for all $i, a_i \in A_i$. We defer the analysis of the
case where players' actions are continuous to
Appendix~\ref{section:continuous}.

\subsection{Perfect counters and empty counters}\label{section:full}

Before delving into our main results, we point out that, with perfect
counters, greedy is the only undominated strategy, and the competitive
ratio of greedy is a constant. We state this result formally, and
defer its proof to Appendix~\ref{app:future-ind-discrete}.

\begin{restatable}{theorem}{greedyfour}\label{thm:greedy4}
  With perfect counters, greedy behavior is dominant-strategy and all
  other behavior is dominated for any sequential resource-sharing game
  $g$; furthermore, $\comp_\greedy(\full, g) = 4$.
\end{restatable}

Recall, from our example in the introduction, that both greedy and
undominated strategies can perform poorly with respect to empty
counters.  We defer the proof of the following results to
Appendix~\ref{app:future-ind-discrete}. Recall that \nnull refers to
the empty counter mechanism.

\begin{restatable}{theorem}{noinfo}\label{thm:noinfo}
  There exist games $g$ such that \small{$\comp_\ud(\nnull, g)$} is unbounded.
\end{restatable}

\begin{restatable}{theorem}{noinfospecial}\label{thm:noinfospecial}
  There exists $g$ such that $\comp_\ud(\nnull, g) \geq
  \Omega(\frac{n^2}{\log(n)})$, when $\dval{t}{r} = \frac{\dval{0}{r}}{t}$.
\end{restatable}

\subsection{Privacy-preserving public counters and Greedy Strategy}\label{section:greedy-counters}
\begin{theorem}
\label{thm:greedy}
With $(\alpha, \beta)$-accurate underestimator counter mechanism $\mech$,
$\comp_{\greedy}(\mech,g) = O(\mult \add)$ for all resource-sharing
game $g$.
\end{theorem}

Before we prove Theorem~\ref{thm:greedy}, we need a way to compare
players' utilities with the utility they \emph{think} they get from
choosing resources greedily with respect to approximate counters. Let
a player's \emph{perceived value} be $\dval{\sused{i}{r}}{r}$ where $r$
is the resource she chose (the value of a resource if the counter was
correct, which may or may not be the \emph{actual} value of the
resource).

\begin{lemma}
\label{lem:perceived}
Suppose players choose greedily according to a $(\mult,
\add)$-underestimator.  Then, the sum of their actual values is at
least a $\frac{1}{2\mult\add}$-fraction of the sum of their perceived
values.
\end{lemma}
\begin{proof}
Suppose $k$ players chose a given resource $r$. For ease of notation, let these be players $1$ through $k$. We wish to bound the ratio \[\frac{\sum_{i=1}^k \dval{\sused{i}{r}}{r}}{\sum_{c=1}^k \dval{c}{r}}.\]
We start by ``grouping'' the counter values: it cannot take on values that are small for more than a certain number of steps. In particular, if $\used{i}{r} > T \alpha \beta$, for some $T \in \mathbb{N}$,
 \begin{align*}
\sused{i}{r} \geq & \frac{1}{\mult}\used{i}{r} - \add  \geq \frac{T \mult \add}{\mult} - \add  = \left(T - 1\right)\add
\end{align*}
Now, we bound the ratio from above using this fact.
\begin{align*}
  \frac{\sum_{i=1}^k \dval{\sused{i}{r}}{r}}{\sum_{c=1}^k \dval{c}{r}}  & \leq \frac{ 2\mult\add \sum_{T = 1}^{\lceil \frac{k}{\mult\add}
    \rceil} \dval{(T -1)\add}{r}}{\sum_{c=1}^k \dval{c}{r}}
  \leq \frac{ 2\mult\add \sum_{T = 1}^{\lceil \frac{k}{\mult\add}\rceil} \dval{(T -1)\add}{r}}{\sum_{T=1}^{\lceil \frac{k}{\mult\add}\rceil} \dval{(T-1)\add}{r}}\leq 2\mult\add
\end{align*}
where the first inequality came from the fact that the value curves
are non-increasing and the lower bound on the counter values from
above, and the second because all terms are nonnegative.
\end{proof}

\begin{proof}[Proof of Theorem~\ref{thm:greedy}]
  The optimal value of the resource-sharing game $g$, denoted by
  $OPT(g)$, is the maximum weight matching in the bipartite
  graph $G = (U \cup V, E)$ where $U$ is the set of the $n$ players
  and $V$ has $n$ vertices for each resource $r$, one of value
  $\dval{k}{r}$ for each $k\in [n]$. There is an edge between player
  $i$ and all vertices corresponding to resource $r$ if and only if $r
  \in A_{i}$. Note that the weights are on the vertices in $V$.

  We now define a complete bipartite graph $G'$ which has the same set
  of nodes but whose node weights differ for some nodes in
  $G$. Consider some resource $r$, and the collection of players who
  chose $r$ in $g$. If there were $t_k$ players $i$ who chose resource
  $r$ when $\sused{i}{r} = k$, make $t_k$ of the nodes corresponding to
  $r$ have weight $\dval{k}{r}$. Finally, if there were $F_k$ players
  who chose resource $r$, let the remaining $n- F_k$ nodes
  corresponding to $r$ have weight $\dval{F_k + 1}{r}$.

  We first claim that the perceived utility of players choosing
  greedily according to the counters is identical to the weight of the
  greedy matching in $G'$ (where nodes arrive in the same order). We
  prove, in fact, that the corresponding matching will be identical by
  induction. Since the counters are monotone, earlier copies of a
  resource appear more valuable. So, when the first player arrives in
  $G'$, the most valuable node she has access to is exactly the first
  node corresponding to the resource she took according to the
  counters. Now, assume that prior to player $i$, all players have
  chosen nodes corresponding to the resource they chose according to
  the counters. By our induction hypothesis and monotonicity of the
  counters and value curves, there is a node $n_i$ corresponding to
  $i$'s selection $r$ according to counters of weight
  $\dval{\sused{i}{r}}{r}$, and no heavier node corresponding to
  $r$. Likewise, for all other resources $r'$, all nodes corresponding
  to $r'$ have weight more than $\dval{\sused{i}{r'}}{r'}$. Thus, $i$
  will take $n_i$ for value $\dval{\sused{i}{r}}{r}$. Thus, the weight
  of the greedy matching in $G'$ equals the perceived utility of
  greedy play according to the counters.

  Let $\greedy_\counters$ denote the set of actions players make
  playing greedily with respect to the counters.  By
  Lemma~\ref{lem:perceived}, the social welfare of $\greedy_\counters$
  is a $\frac{1}{\mult\add}$-fraction of the perceived social
  welfare. By our previous argument, the perceived social welfare of
  greedy play according to the counters is the same as the weight of
  the greedy matching in $G'$. By Theorem~\ref{thm:greedy4}, the
  greedy matching in $G'$ is a $4$-approximation to the max-weight
  matching in $G'$. Finally, since the counters are underestimators,
  the weight of the max-weight matching in $G'$ is at least as large
  as $OPT(g)$. Thus, we know that the social welfare of greedy play
  with respect to counters is a $\frac{1}{2\mult\add}$ fraction of the
  optimal social welfare to $g$.
\end{proof}

\begin{theorem}\label{thm:polylog}
  There exists $(\epsilon, \delta)$-privacy-preserving mechanism \mech
  such that $$\comp_\greedy(\mech, g) = \min\left(O\left(\frac{(\log
      n)(\log (n m /\delta))}{\eps}\right),O\left(\frac{m\log n
      \log\log(1/\delta)}{\eps}\right)\right)$$ for all resource-sharing
  games $g$.
\end{theorem}
\begin{proof}
  In Section~\ref{section:counters}, we prove
  Corollary~\ref{cor:approx} that says that we can achieve an
  $(\eps,\delta)$-differentially private counter vector achieving the
  better of $(1,O(\frac{(\log n)(\log (n m /\delta))}
  {\eps}))$-accuracy and $(\alpha,\tilde O_{\alpha}(\frac{m\log n
    \log\log(1/\delta)}{\eps}))$-accuracy for any constant
  $\alpha>1$. This along with Theorem~\ref{thm:greedy} proves the
  result.
\end{proof}

In Appendix~\ref{sec:moreacc}, Observation~\ref{obs:greedy-apx} proves
that players acting greedily according to any estimate that is
\emph{deterministically} more accurate than the values provided by the
private counters also achieve similar or better social welfare
guarantees. Moreover, we show that if the estimates used by the
players are more accurate only in expectation, as opposed to
deterministically, then we cannot make a similar claim
(Observation~\ref{obs:variance-greedy}).
 
\subsection{Privacy-preserving public counters and Undominated strategies}
\label{section:undominated-counters}
We begin with an illustration of how undominated strategies can
perform poorly for arbitrary value curves, as motivation for the
restricted class of value curves we consider in
Theorem~\ref{thm:undominated}.  In the case of greedy players, we were
able to avoid the problem of players undervaluing resources rather
easily, by forcing the counters to only underestimate
$\used{i}{r}$. This won't work for undominated strategies: players who
know the counts are shaded downward can compensate for that fact.

\begin{theorem}\label{thm:lb-undom}
  For an $(\epsilon, \delta)$--differentially private announcement
  mechanism \mech, there exist games $g$ for which $\comp_\ud(g,
  \mech) = \Omega\left(\frac{1}{\delta}\right)$.
\end{theorem}

\begin{proof}
  Suppose there are two players $1$ and $2$, and resources $r,
  r'$. Let $r$ have $\dval{0}{r} = 1$, $\dval{1}{r} = 0$, and
  $\dval{k}{r'} = \rho$, for all $k \geq 0$. Furthermore, let player
  $1$ have access only to resource $r'$ but player $2$ has access to
  both $r$ and $r'$. Player $1$ will choose $r'$. Let player $2$'s
  strategy be $d_2$, such that if she determines there was nonzero
  chance that player $1$ chose $r$ according to her signal $\ssig{2}$,
  she will choose resource $r'$. This is undominated: if $1$
  \emph{did} choose $r$, $r'$ will be more valuable for $2$.  Thus, if
  $2$ sees any signal that can occur when $r$ is chosen by $1$, she
  will choose $r'$. The collection of signals $2$ can see if $1$
  chooses $r$ has probability $1$ in total. So, because $\ssig{2}$ is
  $(\epsilon, \delta)$-differentially private in player $1$'s action,
  the set of signals \emph{reserved} for the case when $1$ chooses
  $r'$ (that cannot occur when $r$ is chosen by $1$) may occur with
  probability at most $\delta$ (they can occur with probability $0$ if
  $1$ chose $r$, implying they can occur with probability at most
  $\delta$ when $1$ chooses $r'$). Thus, with this probability
  $1-\delta$, player $2$ will choose $r'$, implying $\mathbb{E}[SW]
  \leq (1-\delta) 2 \rho + \delta (1 + \rho) = \delta + (2 -
  \delta)\rho$, which for $\rho$ sufficiently small approaches
  $\delta$, while $1 + \rho$ is the optimal social welfare.
\end{proof}

Given the above example, we cannot hope to have a theorem as general
as Theorem~\ref{thm:greedy} when analyzing undominated strategies with
privacy-preserving counters. Instead, we show that, for a class of
well-behaved value curves, we can bound the competitive ratio of
undominated strategies (Theorem~\ref{thm:undominated}).

Again, along the lines of the greedy case, we show that any player who
chooses any undominated resource $r'$ over resource $r$ gets a
reasonable fraction of the utility she would get from choosing
$r$. Then, by the analysis of greedy players, we have an analogous
argument implying the bound of Theorem~\ref{thm:undominated}. 

\begin{theorem}\label{thm:undominated}
  If each value curve $\dvalo{r}$ has the property that $ \psi(\mult,
  \add) \dval{x}{r} \geq \dval{( \max\{0,\frac{x}{\mult^2} -
    \frac{2\add}{\mult}\})}{r}$ and also $\dval{(\mult^2 x +
    2\mult\add )}{r} \geq \phi(\mult, \add) \dval{x}{r}$, then an
  action profile $a$ of undominated strategies according to $(\mult,
  \add)$-counter vector \mech has $\comp_\ud(g, \mech) =
  O\left(\psi(\mult, \add)\phi(\mult, \add)\right)$.
\end{theorem}

In particular, Theorem~\ref{thm:undominated} shows that, for games
where $\dval{i}{r} = \frac{\dval{0}{r}}{g_r(\used{i}{r})}$, where
$g_r$ is a polynomial, the competitive ratio of undominated strategies
degrades gracefully as a function of the maximum degree of those
polynomials. A simple calculation implies the following corollary, whose
proof we relegate to  Appendix~\ref{section:undominated-independent}.

\begin{corollary}\label{cor:poly}
  Suppose for a resource-sharing game $g$, each resource $r$ has a
  value curve of the form $\dval{x}{r} = \frac{\dval{0}{r}}{g_r(x)}$,
  where $g_r$ is a monotonically increasing degree-$d$ polynomial and
  $\dval{0}{r}$ is some constant. Then, $\comp_\ud(g, \mech) \leq
  O(2\mult^3\beta)^d$ with $\mech$ providing $(\mult, \add)-$counters.
\end{corollary}

\section{Private Counter Vectors with Lower Errors for Small
  Values}\label{section:counters}

In this section, we describe a counter for the model of differential
privacy under continual observation that has improved guarantees when
the value of the counter is small.
Recall the basic counter problem: given a stream $\va =
(a_1,a_2,...,a_n)$ of numbers $a_i \in [0,1]$, we wish to release at
every time step $t$ the partial sum $x_t = \sum_{i=1}^t a_i$.
We require a generalization, where one maintains a vector of $m$
counters. Each player's update contribution is now a vector $a_i \in 
\Delta_m = \{a\in [0,1]^m\mid \|a\|_1\leq 1\}$. That is, a
player can add non-negative values to all counters, but the total value
of her updates is at most 1. The partial sums $x_t$ then lie in
$(\R^+)^m$ and have $\ell_1$ norm at most $t$.

Given an algorithm $\mech$, we define the output stream
$(\Sused{1},...,\Sused{n})=\mech(\va)$ where $\Sused{t} = \mech(t,
a_1,...,a_{t-1})$ lies in $\R^m$. 
We seek counters that are private (Definition~\ref{def:privacy})
and satisfy a mixed multiplicative and additive accuracy guarantee
(Definition~\ref{def:accuracy}). Proofs of all the results in this section can be found in
Appendix~\ref{section:app-counters}.

The
original works on differentially private counters
\cite{dwork2010,ChanSS11} concentrated on minimizing the additive
error of the estimated sums, that is, they sought to minimize
$\|x_t-y_t\|_\infty$. Both papers gave a binary tree-based mechanism,
which we dub ``TreeSum'', with additive error approximately $(\log^2
n)/\eps$. Some of our algorithms use TreeSum, and others use a new
mechanism (FTSum, described below) which gets a better additive error
guarantee at the price of introducing a small multiplicative
error. Formally, they prove:


\begin{lemma}\label{lemma:treesum}
  For every $m\in \mathbb{N}$ and $\pfail\in(0,1)$: Running $m$
  independent copies of TreeSum \cite{dwork2010,ChanSS11} is
  $(\eps,0)$-differentially private and provides an $(1,C_{tree} \cdot
  \frac{(\log n)(\log (n m /\pfail))} { \eps},\pfail)$-approximation
  to partial vector sums, where $C_{tree}>0$ is an absolute constant.
\end{lemma}

Even for $m=1, \alpha=1$, this bound is slightly tighter than those in
\citet{ChanSS11} and \citet{dwork2010}; however, it follows directly
from the tail bound in \citet{ChanSS11}.

Our new algorithm, FTSum (for Flag/Tree Sum), is described in
Algorithm~\ref{alg:mixedcounter}. For small $m$ ($m = o(log(n))$), it
provides lower additive error at the expense of introducing an
arbitrarily small constant multiplicative error.

\begin{lemma}\label{lemma:FTSum}
  For every $m\in \mathbb{N}$, $\alpha>1$ and $\pfail\in (0,1)$, FTSum
  (Algorithm \ref{alg:mixedcounter}) is $(\eps,0)$-differentially
  private and $(\alpha,\tilde O_{\alpha}(\frac{m\log (n
    /\pfail)}{\eps}), \pfail)$-approximates partial sums (where
  $\tilde O_a(\cdot)$ hides polylogarithmic factors in its argument,
  and treats $\alpha$ as constant).
\end{lemma}


FTSum proceeds in two phases. In the first phase, it increments the
reported output value only when the underlying counter value has
increased significantly. Specifically, the mechanism outputs a public
signal, which we will call a ``flag'', roughly when the true counter
achieves the values $\log n$, $\alpha\log n$, $\alpha^2\log n$ and
so on, where $\alpha$ is the desired \emph{multiplicative}
approximation. The reported estimate is updated each time a flag is
raised (it starts at 0, and then increases to $\log n$, $\alpha\log
n$, etc). The privacy analysis for this phase is based on the
``sparse vector'' technique of \citet{HR10}, which shows that the cost
to privacy is proportional to the number of times a flag is raised
(but not the number of time steps between flags).  

When the value of the counter becomes large (about $\frac{\alpha
  \log^2 n}{(\alpha-1)\eps}$), the algorithm switches to the second
phase and simply uses the TreeSum protocol, whose additive error
(about $\frac{\log^2 n}{\eps}$) is low enough to provide an $\alpha$
multiplicative guarantee (without need for the extra space given by
the additive approximation).

If the mechanism were to raise a flag \emph{exactly} when the true
counter achieved the values $\log n$, $\alpha\log n$, $\alpha^2\log
n$, etc, then the mechanism would provide a $(\alpha,\log n,0)$
approximation during the first phase, and a $(\alpha,0,0)$
approximation thereafter. The rigorous analysis is more complicated,
since flags are raised only near those thresholds.

\newcommand{\flag}{\text{flag}}

\begin{algorithm}\label{alg:mixedcounter}
  \caption{FTSum --- A Private Counter with Low Multiplicative Error}
  \KwIn{Stream $\va = (a_1,...,a_n) \in ([0,1]^m)^n$, parameters
    $m,n\in \en$, $\alpha>1$ and $\pfail >0$}
  \KwOut{Noisy partial sums $y_1,...,y_n\in \re^m$}
  $k\gets \lceil\log_{\alpha}(\frac{\alpha}{\alpha -1}\cdot
  C_{tree}\cdot \frac{\log (nm/\gamma)}{\eps})\rceil$\;
\tcc{{\color{blue}$C_{tree}$ is the
    constant from Lemma~\ref{lemma:treesum}}}
  $\eps' \gets \frac{\eps}{2m(k+1)}$\; 
  \For{{$r=1$ \KwTo $m$}}
  {
    $\flag_r \gets 0$\;
    $\used{0}{r}\gets 0$\;
    $\tau_r \gets (\log n)+ \Lap(2/\eps')$\;
  }
  \For{$i=1$ \KwTo $n$}
  {
    \For{$r=1$ \KwTo $m$}
    {
      \If({({\color{blue} First phase still in
          progress for counter $r$})}){$\flag_r \leq k$}
      {
        $\used{i}{r}\gets \used{i-1}{r}+\ract{i}{r}$\;
        $\tilde {\used{i}{r}} \gets \used{i}{r} + \Lap(\frac{2}{\eps'})$\;
        \If({({\color{blue} Raise a new flag for counter $r$})}){$\tilde {\used{i}{r}} >\tau_r$} 
        {
          $\flag_r\gets \flag_r +1$\;
          $\tau_r \gets (\log n)\cdot \alpha^{\flag_r}+ \Lap(2/\eps')$\;
        }
        \textbf{Release} $\sused{i}{r} = (\log n)\cdot
        \alpha^{\flag_r-1}$ \;
      }
      \Else({({\color{blue} Second phase has been reached for  counter
        $r$})})
      {
        \textbf{Release} $\sused{i}{r} =$ $r$-th counter output from TreeSum$(\va,\eps/2)$)\;
      }
    }
  }
\end{algorithm}

\paragraph{Enforcing Additional Guarantees}
Finally, we note that it is possible to enforce to additional useful
properties of the counter. First, we may insist that the accuracy
guarantees be satisfied with probability 1 (that is, set $\gamma=0$),
at the price of increasing the additive term $\delta$ in the privacy guarantee:

\begin{prop}
\label{prop:zero-failure}
  If $\mech$ is  $(\eps,\delta)$-private and $(\alpha,\beta,\gamma)$-accurate, then one can modify $\mech$ to obtain an algorithm $\mech'$ with the same efficiency that is $(\eps,\delta+\gamma)$-private and
  $(\alpha, \beta,0)$-accurate.
\end{prop}

Second, as in
\cite{dwork2010}, we may enforce the requirement that the reported values be
monotone, integral values that increase at each time step by at most
1. The idea is to simply report the nearest integral, monotone
sequence to the noisy values (starting at 0 and incrementing the
reported counter only when the noisy value exceeds the current
counter). 

\begin{prop}[\cite{dwork2010}]
\label{prop:monotone}
  If $\mech$ is  $(\eps,\delta)$-private and
  $(\alpha,\beta,\gamma)$-accurate, then one can modify $\mech$ to obtain
  an algorithm $\mech'$ which reports monotone, integral values that
  increase by 0 or 1 at each time step, with the same privacy and
  accuracy guarantees as $\mech$.
\end{prop}

\begin{corollary}\label{cor:approx}
  Algorithm~\ref{alg:mixedcounter} is an $(\eps,\delta)$-differentially
  private vector counter algorithm providing a
  \begin{enumerate}
  \item $(1,O(\frac{(\log n)(\log (n m /\delta))} {
   \eps}),0)$-approximation (using modified TreeSum); or
\item $(\alpha,\tilde O_{\alpha}(\frac{m\log n
    \log\log(1/\delta)}{\eps}), 0)$-approximation for any
  constant $\alpha>1$ (using FTSum).
  \end{enumerate}
\end{corollary}


\section{Extensions}\label{sec:extensions}
As part of Appendix~\ref{section:future-appendix}, we also consider settings where players' utility when choosing a resource $r$ depends upon the \emph{total number} of players choosing $r$, not just the  players who chose $r$ before. In addition, Appendix~\ref{sec:other-games} considers several other classes of games: namely, cut games, consensus games, and unrelated machine scheduling, and consider whether or not private synopses of the state of play is sufficient to improve social welfare over simultaneous play, as perfect synopses have been proven to be in~\citet{Leme:2012}.

\section{Discussion and Open Problems}\label{section:discussion}
In this work, we considered how public dissemination of information in
sequential games can guarantee a good social welfare while maintaining
differential privacy of the players' strategies. We considered two
`extreme' cases -- the greedy strategy and the class of all
undominated strategies. While analyzing the class of undominated
strategies gives guarantees that are robust, in many games that we
considered, the competitive ratios were significantly worse than
greedy strategies, and in some cases they were unbounded. It is interesting to note that many of the examples in this paper that show the poor performance with undominated strategies also hold when the players know their position in the sequence, an assumption we have not made for any of the positive results in this work. It is an
interesting direction for future research to consider classes of
strategies that more restricted than undominated strategies yet are
general enough to be relevant for games where players play with
imperfect information.

As mentioned in the introduction, we note here that, while players are
making choices subject to approximate information, our results are not
a direct extension of the line of thought that approximate information
implies approximate optimization. In particular, for greedy
strategies, while there may be a bound on the error of the counters,
but that \emph{does not imply} that, for arbitrary value curves,
playing greedily according to the counters will be \emph{approximately
  optimal for each individual}. In particular, consider one resource
$r$ with value $H$ for the first $10$ investors, and value $0$ for the
remaining investors, and a second resource $r'$ with value $H/2$ for
all investors. With $(\mult, \add, \pfail)$, as many as $\add$ players
might have unbounded ratio between their value for $r$ as $r'$, but
will pick $r$ over $r'$. The analysis of greedy shows, despite this
anomaly, the total social welfare is still well-approximated by this
behavior.

All of our results relied on using differentially private counters for
disseminating information. For the differentially-private counter, a
main open question is ``what is the optimal trade-off between additive
and multiplicative guarantees?". Furthermore, as part of future
research, one can consider other privacy techniques for announcing
information that can prove useful in helping players achieve a good
social welfare. And more generally, we want to understand what
features of games lend themselves to be amenable to public
dissemination of information that helps achieve good welfare and
simultaneously preserves privacy of the players' strategies.

\small{
\bibliography{sources,other}}
\appendix
\section{Other games}
\label{sec:other-games}
In this section, we study a number of games which~\citet{Leme:2012} showed to have a large improvement between their Price of Anarchy and their \emph{sequential} Price of Anarchy. We pose the question: with privacy-preserving information handed out to players, what loss is incurred in comparison to providing exact information? In addition to introducing privacy constraints, we should note here that while in \citet{Leme:2012}, each player playing the sequential game knows the \emph{type of every other player}, in our setting, we only provide information about \emph{actions taken by previous players}.

\subsection{Unrelated Machine scheduling
  games}\label{section:machine-scheduling}
An instance of the unrelated machine scheduling game consists of $n$
players who must schedule their respective jobs on one of the $m$
machines; the cost to the player is the final load on the machine on
which she scheduled her job. The size of player $k$'s job on machine
$q$ is $t_{kq}$. The objective of the mechanism is to minimize the
makespan. We consider the dynamics of this game when it is played
sequentially. \citet{Leme:2012} prove that the sequential price of
anarchy is $O(m2^{n})$.

In our setting, each player is shown a \emph{load profile} when it is
her turn to play. The \emph{load profile} $L$ denotes the
\emph{displayed} vector of loads on the various machines. In the
\emph{perfect counter} setting, the displayed load equals the exact
load on each machine. We now show that undominated strategies, with
perfect counters, perform unboundedly poorly with respect to OPT.

\begin{lemma}
  If $\mech$ is a perfect counter vector, $CR_\ud(\mech, g)$ is
  unbounded for some instances $g$ of unrelated machine scheduling.
\end{lemma}
\begin{proof}
  Consider the case with two players ($p_{1}$ and $p_{2}$) and two
  machines ($m_{1}$ and $m_{2}$). $p_{1}$ arrives before
  $p_{2}$. Player $p_{1}$ has a cost of 0 on $m_{1}$ and 1 on
  $m_{2}$. It is an undominated strategy for player 1 to choose
  $m_{2}$ since if player $p_{2}$ has a cost of 2 on $m_{1}$ and 3 on
  $m_{2}$, $p_{2}$ chooses $m_{1}$ and so $p_{1}$ is better off
  scheduling her job on $m_{2}$.

  However, if $p_{1}$ chooses $m_2$ (an undominated strategy) $m_{2}$,
  and player $p_{2}$ has cost 1 on $m_{1}$ and 0 on $m_{2}$, the
  optimal makespan is 0; the achieved makespan is at least 1.
\end{proof}

In light of this result, we restrict our attention to greedy
strategies for machine scheduling, and show that the competitive ratio
of the greedy strategy with privacy-preserving counters is
bounded. Below, we denote by $t^{*}_{k}$ the minimum cost of job $k$ among all the machines and by $q^{*}_{k}$ the machine that achieves this minimum. The following result follows from the analysis
of the greedy algorithm as presented in~\citet{Aspnes:1997}.

\begin{theorem}~\citep{Aspnes:1997} With perfect counters and players
  playing greedy strategies, the makespan is at most
  $\sum_{i=1}^{n}t_{i}^{*}$, and since $OPT \ge
  \sum_{i=1}^{n}t_{i}^{*}/m$, the competitive ratio is at most $m$.
\end{theorem}

Theorem~\ref{thm:scheduling-greedy} shows that such a bound extends to
the setting where players have only approximate information about the
state, showing that privacy-preserving information is enough to attain
nontrivial coordination with greedy players.

\begin{theorem}\label{thm:scheduling-greedy}
Using $(\mult, \add, \pfail)$-counter vector, and players playing greedy strategies, with probability $1-\pfail$, the makespan is at most $\mult^{2n+1} m \cdot OPT + \add(\mult^{2n+1}(2n+1) + 1)$.
\end{theorem}

\begin{proof}[Proof of Theorem~\ref{thm:scheduling-greedy}]
  Consider any player $i$, and let the \emph{displayed} load profile
  she sees be $L$. Using greedy strategy, she will put her job on
  machine $q$ that minimizes $L_{q} + t_{kq}$ and this in particular
  shall be at most $\norminfty{L} + t^{*}_{k}$. Since the \emph{true}
  makespan before this player placed her job is at most $\mult
  \norminfty{L} + \add$, hence after she places her job, for the
  displayed load profile $L'$, $\norminfty{L'} \le \mult (\mult
  \norminfty{L} + \add + t^{*}_{k})+\add \le \mult^{2}(\norminfty{L} +
  2\add + t^{*}_{k})$.

  Using the above reasoning for every player in the sequence, we have
  the displayed load profile $L_{n}$ at the end of the sequence has
  the property that $\norminfty{L_{n}} \le
  \mult^{2n}(\norminfty{L_{0}} + 2n\add + \sum_{k=1}^{n}t^{*}_{k})$,
  where $L_{0}$ is the load profile shown to the first player. But
  $\norminfty{L_{0}}$ is at most $\add$, since the true load on all
  machines is zero at that point.

  Since the displayed makespan at the end of the sequence is at most
  $\mult^{2n}(\add+ 2n\add + \sum_{k=1}^{n}t^{*}_{k})$, hence the true
  makespan is at most $\mult^{2n+1}(\add + 2n\add +
  \sum_{k=1}^{n}t^{*}_{k})+ \add$. Since $OPT \ge \sum_{k=1}^{n}t^{*}_{k}/m$ we have our result.
\end{proof}

\subsection{Cut games}\label{section:cut-games}
A cut game is defined by a graph, where every player is a node of the
graph. Each of the $n$ players chooses one of the two colors, `red' or
`blue', and the utility to a player is the number of her neighbors who
\emph{do not have} the same color as hers.

In sequential play, when a player has her turn to play, she is shown
counts of the number of her neighbors who are colored `red' and who
are colored `blue'. We assume each player knows the total number of
her neighbors in the graph exactly. With greedy strategies, each
player chooses the color with fewer nodes when it is her turn to play.
As was the case for machine scheduling, undominated strategies for cut
games perform much worse than $OPT$, even with perfect counters.

\begin{lemma}\label{thm:cut-game-ud-full}
With perfect counters and undominated strategies, the competitive ratio against the optimal social welfare is at least $n$.
\end{lemma}
\begin{proof}
  Consider the graph to be a long cycle with $2n$ nodes. For ease of  analysis, number the nodes 0 through $2n-1$ with the node numbered
  $i$ have its neighbors $(i-1)\mod 2n$ and $(i+1)\mod 2n$. The
  optimal social welfare is $4n$ obtained by coloring all even
  numbered nodes with red and the rest with blue.

  Consider the sequence of nodes where nodes arrive in the increasing
  order $0$ through $2n-1$. We claim that through a series of
  undominated strategy plays on part of each player, the coloring
  where node $2n-1$ is colored red and the rest colored blue is
  achievable. Note that this coloring gives a social welfare of 4.

  We now prove our claim. It is an undominated strategy for node 0 to
  choose the color blue. Node 1 sees one of its neighbors colored blue
  and the other uncolored. It is an undominated strategy for node 1 to
  choose color blue as well. This continues and each node until node
  $2n-1$ is colored blue. Node $2n-1$ has both its neighbors colored
  blue, and so the only undominated strategy for her is to play red.
\end{proof}

Given the previous result, we focus our attention on greedy strategies. With greedy strategies and perfect counters, the competitive ratio is constant, shown by~\citet{Leme:2012}. We show that, with privacy-preserving counters, it is possible to compare the social welfare of greedy to that of $OPT$.

\begin{theorem}\citep{Leme:2012}
\label{thm:cut-game-greedy-full}
With perfect counters and greedy strategies, the competitive ratio against the optimal social welfare is at most $2$. 
\end{theorem}
\begin{proof}
  Consider the choice made by player $t$ when it is her turn to
  play. Let $C_{t}$ be the number of neighbors of player $t$ that have
  adopted a color by the time it her turn to play. Notice that the
  total number of edges in the graph is $\sum_{t}C_{t}$. Furthermore,
  the greedy strategy ensures that player $t$ gets value at least
  $C_{t}/2$. Since the number of edges in the graph is an upper bound
  on the optimal social welfare, hence we have the greedy strategy
  achieving a competitive ratio of $2$.
\end{proof}

Now, we compare the performance of greedy w.r.t. to approximate counters to $OPT$.

\begin{theorem}\label{thm:cut-game-greedy-private}
  With $(\mult, \add, \pfail)$-counter vector and greedy strategies, with probability at least $1-\pfail$, the social welfare is at least $\frac{OPT}{2\mult^{2}} - \frac{2\add}{\mult}n$.
\end{theorem}
\begin{proof}
  Let us analyze the play made by player $t$ when it is her turn to
  play. Let $R_{t}$ and $B_{t}$ be the true counts of red and blue
  neighbors of $t$ at that time, and without loss of generality let
  $R_{t}\ge B_{t}$. Either the player chooses the blue color and this
  guarantees her utility of $C_{t}/2$, where $C_{t} = R_{t} +
  B_{t}$. On the other hand, if the player were to choose the color
  red, it must be the case that the displayed value of the blue
  counter is at least the displayed value of the red counter. For this
  to be true, it must be the case that $\alpha B_{t} + \beta \ge
  R_{t}/\alpha - \beta$, and therefore $B_{t} \ge R_{t}/\alpha^{2} -
  2\beta/\alpha \ge C_{t}/(2\alpha^{2}) - 2\beta/\alpha$. Hence, in
  either case, the player achieves utility of at least
  $C_{t}/(2\alpha^{2}) - 2\beta/\alpha$.

  Following the analysis used in the proof of
  Theorem~\ref{thm:cut-game-greedy-full}, we have the result.
\end{proof}

\subsection{Cost sharing games}\label{section:cost-sharing}
A cost sharing game is defined as follows. $n$ players each have to
choose one of the $m$ sets. There is an underlying bipartite graph
between the players and the sets, and a player can choose only one
among those sets that she is adjacent to (i.e., she shares an edge
with). Moreover, every set $i$ has a cost $c_{i}$ and the cost to a
player is the cost of the set she chooses divided by the number of
players who chose that set i.e., each of the players who choose a
particular set share its cost equally. Each player would like to
minimize her cost; the social welfare is the sum of costs of the
players, which is equal to the sum of the costs of the sets chosen by
various players.

\citet{Leme:2012} prove that the sequential price of anarchy is
$O(\log(n))$.  Our work uses counters to publicly display an
estimate of the number of the players who have selected that set so
far. With perfect counters, this estimate is always exact.
Unfortunately, greedy strategies can perform poorly in this setting,
even with exact counters.

\begin{lemma}
\label{lem:cost-sharing-perfect-greedy}
With perfect counters and greedy strategies, the competitive ratio is $n$.
\end{lemma}
\begin{proof}
  We first show that the competitive ratio is at most $n$. Let
  $\hat{s}_{i}$ be the set that $i$ should choose in the optimal
  allocation, and let $s_{i}$ she chose. Also, let $l(s_{i})$ be the
  number of player who chose set $s_{i}$. Greedy strategy dictates
  that it must be the case that $c_{s_{i}}/l(s_{i}) \le
  c_{\hat{s}_{i}}$. Summing over all players $i$, we have the total
  cost of the allocation produced by the mechanism is
  $\sum_{i=1}^{n}c_{s_{i}}/l(s_{i}) \le
  \sum_{i=1}^{n}c_{\hat{s}_{i}}$, and this is equal to $\sum_{j \in
    J}q(j)c_{j}$, where $J$ is the collection of sets picked in the
  optimal allocation and $q_{j}$ is the number of players allocated to
  set $j$. Since the optimal cost is $\sum_{j \in J}c_{j}$ and $q_{j}
  \le n$, we have the competitive ratio is at most $n$.

  We now show that the competitive ratio is at least $n$. Consider the
  case where there is a public set $s$ that is adjacent to all the
  players and has cost $1+\epsilon$ (for any small $\epsilon > 0$). In
  addition, there are $n$ private sets $s_{1}, \cdots, s_{n}$ with set
  $s_{i}$ having cost 1 and adjacent only to player $i$. In the
  sequential game play, with greedy strategies and perfect counters
  (indicating the number of players who have chosen a particular set
  so far in the game), each player will choose her private set since
  that will have cost 1 as opposed to $1+\epsilon$ for the public
  set. This gives a total cost of $n$. The optimal solution is to pick
  the public set with a total cost of $1+\epsilon$.
\end{proof}

In light of Lemma~\ref{lem:cost-sharing-perfect-greedy}, the greedy
strategy with respect to approximate counters should not perform well
with respect to $OPT$. However, we do show that there are instances in
which greedy with respect to these approximate counters can be better
than greedy with respect to perfect counters. The example we use is
the same as in Lemma~\ref{lem:cost-sharing-perfect-greedy}, and is
also to the example showing the price of anarchy for cost-sharing is
$\Omega(n)$. Proposition~\ref{prop:private-better-than-perfect} and
the exponential improvement of the sequential price of anarchy over
the simultaneous price of anarchy~\citep{Leme:2012} suggest the
instability of this equilibrium.

\begin{prop}
\label{prop:private-better-than-perfect}
In certain instances of cost sharing with greedy strategies, the
competitive ratio using privacy-preserving counters is better than
using perfect counters.
\end{prop}
\begin{proof}
  Consider the same instance as in
  Lemma~\ref{lem:cost-sharing-perfect-greedy}. There is a public set
  that is adjacent to all the players and has cost $1+\epsilon$. In
  addition, there is a private set for each player that is adjacent to
  only that player. Each private set has cost 1. The number of players
  is $n$ and the number of sets is $m=n+1$.

  Consider the following construction of the counter vector (here
  $p=1$, $q=O(\log(n)\log(n^{2}m)/\epsilon)$, $r=1/n$ and
  $c=8(p^{2}+2p q)$). For the initial sequence of $c$ players, for
  each player $i\in [c]$, for each counter, a uniformly randomly
  chosen number in the range $[0, c]$ (drawn independently for each
  counter) is displayed. Starting with the $(c+1)$st player, each
  counter in the counter vector displays the value according to
  $(p,q,r)$--Tree-sum based construction
  (Lemma~\ref{lemma:treesum}). It is easy to verify that the
  construction gives a $(\mult, \add, \pfail)$ counter vector for
  $\mult=p$, $\add=c$ and $\pfail=r$.

  Let $P$ be the counter that corresponds to the public set, and
  $S_{i}$ be the counter for the $i^{th}$ private set in the counter
  vector. Initially, the true value of all the counters is 0. For the
  initial set of $c$ players, for each $i\in [c]$, the probability
  that the displayed value of $P$ is greater than that of $S_{i}$ is
  $1/2$ (since for each player $i \in [c]$, on each counter, a
  uniformly random number drawn independently from the range $[0,c]$
  is displayed).
 
  Hence, in the first $c$ players, the expected number of players for
  who the displayed value of $P$ is greater than the corresponding
  $S_{i}$ is $c/2$, and under greedy strategy, all these players will
  choose the public set. Hence the expected true count of the $P$ at
  the end of the prefix of $c$ players is $c/2$. Using a Chernoff
  bound, the probability the true count of $P$ after the first $c$
  players is smaller than $c/4$ is at most $e^{(-c/16)}$.

  After the initial sequence of $c$ players, the counter values are
  displayed according to the $(p,q,r)$-Tree based construction. By the
  error guarantees, it follows that if the true count for the public
  set is at least $p^{2} + 2p q$ at the end of the initial $c$-length
  sequence, then for the rest of the players, with probability
  $(1-r)$, the displayed value of $P$ is \emph{always strictly
    greater} than the displayed value of every $S_{i}$ (whose true
  value is at most 1 and so the displayed value is at most $p +
  q$). Since $c/4= 2(p^{2} + 2p q)$, we can infer that with
  probability at least $(1-r -e^{(-c/16)})$, all players after the
  initial sequence of length $c$ will choose the public set giving the
  total cost of at most $1+\epsilon + c$. In contrast, with perfect
  counters, the total cost is always $n$
  (Lemma~\ref{lem:cost-sharing-perfect-greedy}).
\end{proof}

\section{Future Independent: Discrete Version}
\label{app:future-ind-discrete}

\greedyfour*

The proof of this Theorem follows from the connection between
future-independent resource-sharing and online vertex-weighted
matching, which we mention below.

\begin{obs}\label{obs:match}
  In the setting where $\|a_i\|_1 = 1$ for all $a_i\in A_i$, for all
  $i$, full-information, discrete resource-sharing reduces to online,
  vertex-weighted bipartite matching.
\end{obs}

\begin{proof}
  Construct the following bipartite graph $G = (U, V, E)$ as an
  instance of online vertex-weighted matching from an instance of the
  future-independent resource sharing game. For each resource $r$,
  create $n$ vertices in $V$, one with weight $\dval{t}{r}$ for each
  $t\in [n]$. As players arrive online, they will correspond to
  vertices in $u_i \in U$.  For each $a_i \in A_i$ corresponding to a
  set of resources $S$, $u_i$ is allowed to take any subset of $V$
  with a single copy of each $r\in S$. 
\end{proof}

The proof of the social welfare is quite similar to the one-to-one,
online vertex-weighted matching proof of ~\citep{kvvmatching}, with
the necessary extension for many-to-one matchings (losing a factor of
$1/2$ in the process).

\begin{proof}[Proof of Theorem~\ref{thm:greedy4}]
  Consider any instance of $G = (U, V, E)$, a vertex-weighted
  bipartite graph. Let $\mu$ be the optimal many-to-one matching,
  which can be applied to nodes in both $U$ and $V$ (where $u\in U$
  has potentially multiple neighbors in $V$).  Consider $\mu'$, the
  greedy many-to-one matching for a particular sequence of arrivals
  $\sigma$.

  Consider a particular $u\in U$, and the time it arrives $\sigma(u)$
  as $\mu'$ progresses.  If at least $1/2$ the value of $\mu(u)$ is
  available at that time, then $w(\mu'(u)) \geq \frac{1}{2}w (\mu(u))$
  (since $u$ can be matched to any subset of $\mu(u)$, by the downward
  closed assumption). If not, then $w(\mu'(\mu(u))) \geq \frac{1}{2} w
  (\mu(u))$ (at least half the value was taken by others). Thus, we
  know that, for all $u$,

\[ w(\mu'(u)) + w(\mu'(\mu(u))) \geq \frac{1}{2} w (\mu(u))\]

summing up over all $u$, we get

\begin{align*}
\sum_u w(\mu'(u)) + w(\mu'(\mu(u)))= 2 w(\mu') \geq \frac{1}{2} \sum_u w (\mu(u)) = \frac{1}{2} w(\mu)
\end{align*}

Rearranging shows that $w(\mu') \geq \frac{1}{4} w(\mu)$.

Finally, the utility to a player is clearly greatest when they are
greedy, so that is a dominant strategy (thus implying any non-greedy
strategy is dominated).
\end{proof}

\subsection{Greedy play with more accurate estimates}\label{sec:moreacc}
\begin{obs}\label{obs:greedy-apx}
  Suppose that \mech is a $(\mult, \add, \pfail)$-underestimating counter vector, giving estimates
  $\sused{i}{r}$. Furthermore, assume each player $i$ is playing
  greedily with respect to a revised estimate $z_{i,r}$ such
  that, for each $r,i,$ and value of $z_{i}{r}$ is always in the range $[\sused{i}{r}, \used{i}{r}]$. Then, for $g$, a discrete
  resource-sharing game, with probability $1-\pfail$, the ratio of the optimal to the achieved social welfare is $O(\mult\add)$.
\end{obs}
\begin{proof}
  The proof follows from the proof of Theorem~\ref{thm:greedy}, along with the following observation. Since $z_{i,r}$'s is deterministically more accurate than the \counters, we have for each $i$ that the value gained by greedily choosing according to the estimates  $z_{i,r}$ is at least as much as the value gained by greedily choosing using $\sused{i}{r}$. Therefore, summing over all the players, the achieved social welfare is at least as much as it would be if everyone had played greedily according to $\sused{i}{r}$.
\end{proof}

\begin{obs}\label{obs:variance-greedy}
  There exists a resource-sharing game $g$, such that if the players play greedily according to estimates $z_{i,r}$ that are more accurate than the displayed value only in expectation -- specifically for each $r,i,$ and value  of $\used{i}{r}$, $\prob[z_{i,r} < \used{i}{r}] \geq 1/2$ and also  $\exp{|z_{i,r} - \used{i}{r}|} = 1$, then the ratio of the optimal to the achieved social welfare can be as bad as $\Omega\left(\sqrt{n}\right)$.
\end{obs}
\begin{proof}
  Let there be $n + \sqrt{n}$ resources, with resources $r*_{1,
    \ldots, \sqrt{n}}$ having $\dval{0}{r*_f} = H$, $\dval{t}{r*_f} =
  0$ for all $t > 0$, and resource $r_i$ such that $\dval{t}{r_i} =
  H-\epsilon$ for all $t$. Player $i$ has access to all resources
  $r*_f$ and $r_i$. Then, $OPT = H\sqrt{n} + (H-\epsilon)(n -
  \sqrt{n}) = H n - (n - \sqrt{n})\epsilon$. 

  Consider the counter vector which is exactly correct with
  probability $1-\frac{1}{\sqrt{n}}$ and undercounts by $\sqrt{n}$
  with probability $\frac{1}{\sqrt{n}}$ (note that the expected error
  is just $1$ and it undercounts with probability 1). Then, greedy
  behavior with respect to this counter will (in expectation) have
  $\sqrt{n}$ players choose $r*_f$ for each $f$, achieving welfare
  $\sqrt{n}H$. Thus, the competitive ratio is $\Omega(\sqrt{n})$ as
  $\epsilon\to 0$, as desired.
\end{proof}

\subsection{Undominated strategic play with Empty Counters: Lower
  bounds}\label{section:noinfo}
\noinfo*

\begin{proof}
  Let $g$ be the following game. For each player $i$, there is a resource $r_i$ such that $\val{1}{r_i} = H$ but $\val{> 1}{r_i} =
  0$. Furthermore, let there be some other resource $r$ such that
  $\val{1}{r} = 1$. Let $A_i$ contain 2 allowable actions: selecting
  $r_i$ and selecting $r$.

  OPT in this setting would have each player select $r_i$, which has
  $SW(OPT) = n H$. On the other hand, we claim it is undominated for
  each player to select $r$ instead (call this joint action $a$). If
  each player were to have a ``twin'', then $r_i$ could have already
  been selected by another player so that $i$ would get more utility
  from $r$ than $r_i$. Then, this undominated strategy $a$ has $SW(a)
  = n$. Thus, we have a game $g$ for which

\[\comp_\ud(g) \geq \frac{n H}{n} = H\]

which, as $H\to \infty$ is unbounded.
\end{proof}

The negative result above isn't particularly surprising: if there is
some coordination to be done, but there is no coordinator and no
information about the target, all is lost. On the other hand, our
positive result for undominated strategies
(Theorem~\ref{thm:undominated}) in the case of private information
relies on a very particular rate of decay of the resources' value.
Theorem~\ref{thm:noinfospecial} show that, even under this stylized
assumption where all resources' values shrink slowly, a total lack of
information can lead to very poor behaviour in undominated strategies.

\noinfospecial*

\begin{proof}[Proof of Theorem~\ref{thm:noinfospecial}]
  For each player $i$, let $r_i$ be a resource where $\val{1}{r_i} =
  n$ (note that this uniquely determines $\val{c}{i}$ for all
  $c$). Let there be another resource $r$ such that $\val{1}{r} =
  1$. Let each $A_i$ contain all resources. Since
  $\frac{\val{1}{r_i}}{n} = 1$, it is not dominated for player $i$ to
  select $r$. Let $a$ denote the joint strategy where each player
  selects resource $r$. Thus the social welfare attained by this strategy profile is $O(\log(n))$, where as the optimal social welfare is $n^2$, implying that $\comp_\ud\geq \Omega(\frac{n^2}{\log(n)})$.
\end{proof}

\subsection{Omitted proofs for Undominated strategies with Privacy-preserving counters}
\label{section:undominated-independent}
\begin{proof}[Proof of Theorem~\ref{thm:undominated}]
Consider the optimal allocation and let $r_{i}$ and $z_{i}$ denote that the $z_{i}^{th}$ copy of resource $r_{i}$ got allocated to player $i$ under the optimal allocation. Now consider any run of the game under undominated strategic play and based on the run, partition all the players into two groups. Group $A$ consists of players $i$ such that $\used{i}{r_{i}} \le z_{i}$ (i.e., the copy (or a more valuable copy) of the resource that was allocated to player $i$ was present when the player arrived) and group $B$ consists of all other players.

For the player in group $B$, the copy of the resource that they received in the optimal allocation was already allocated by the time they arrived in the run of the undominated strategic play. Hence, the total social welfare achieved by the undominated strategic play is at least as much the welfare achieved by group $B$ player under optimal allocation.

Now consider any player $i$ in group $A$. We show that the resource picked by player $i$ under undominated strategic play gets her a reasonable fraction of the value she would have received under optimal allocation.  For any resource $r$, given the displayed counter value of $\sused{i}{r}$, by the guarantees of the $(\alpha, \beta)$-accuracy guarantee of the counters, we directly argue about the possible range of the consistent beliefs or estimates $\hat{\used{i}{r}}$ by which player $i$ can make her choice.

Specifically, by the bounds on $(\mult, \add)$-counters, for a given true value $x$, it must be the case that all announcements $\sused{i}{r}$ satisfy:
\[ \alpha \used{i}{r} + \beta \geq \sused{i}{r} \geq \frac{1}{\alpha}\used{i}{r} - \beta\]

Rearranging, we have $\sused{i}{r} \in [ \frac{1}{\mult}\used{i}{r} -
\add, \mult\used{i}{r} + \add]$. Suppose these bounds are realized; we
wish to upper and lower bound $\hat{\used{i}{r}}$ as a function of
these announcement values. By the quality of the announcement, we have
that $\mult \hat{\used{i}{r}} + \beta \geq \sused{i}{r} \geq
\frac{1}{\mult}\used{i}{r} - \add$.

We can similarly upper bound $\hat{\used{i}{r}}$, e.g.  $ \mult
\used{i}{r} + \beta \geq \sused{i}{r}\geq
\frac{1}{\mult}\hat{\used{i}{r}} - \add$, which, by the fact that the
true count is at least 0, implies $\hat{\used{i}{r}} \in [
\max\{0,\frac{\used{i}{r}}{\mult^2} - \frac{2\add}{\mult}\}, \mult^2
\used{i}{r} + 2\mult\add].$ 

Now, suppose player $i$ chose resource $r'$ which was undominated and not $r_{i}$ which he received in the optimal allocation. Since resource $r'$ is undominated:
\begin{align}
\dval{\hat{\used{i}{r'}}}{r'} \geq \dval{\hat{\used{i}{r_{i}}}}{r_{i}} \label{eqn:undom}
\end{align}

We have
\begin{align}
\dval{\hat{\used{i}{r'}}}{r'} \leq \dval{\max\{0,\frac{\used{i}{r'}}{\mult^2}}{r'} - \frac{2\add}{\mult}\}) \leq \psi(\mult, \add) \dval{\used{i}{r'}}{r'}\label{eqn:upper}
\end{align}
where the first inequality came from the lower bound on the counter, and the fact that the valuations are decreasing, and the second from the assumption about $\dvalo{r}$ on $x$ and its lower bound.
Similarly, we know for each $r$ that
\begin{align}
\dval{\hat{\used{i}{r_{i}}}}{r_{i}} \geq \dval{ \mult^2 \used{i}{r_{i}} + 2\mult\add}{r_{i}} \geq \frac{\dval{\used{i}{r_{i}}}{r_{i}}}{\phi(\mult, \add)} \label{eqn:lower}
\end{align}

Combining the three equations above, we have the actual value received by the player $i$ on choosing resource $r'$, $\dval{\used{i}{r'}}{r'}$ is at least $\frac{1}{ \psi(\mult, \add)\phi(\mult, \add)}$ fraction of the value $\dval{\used{i}{r_{i}}}{r_{i}}$ that he would receive under the optimal allocation. Therefore, by virtue of partition of the players in groups $A$ and $B$, we have that social welfare achieved under undominated strategic play is at least $\frac{1}{1 + \psi(\mult, \add)\phi(\mult, \add)}$ fraction of the optimal social welfare.
%
%
%
\end{proof}



\section{Resource Sharing: Continuous version}
\label{section:continuous}
\label{section:cont}
In this section, we allow investments in resources to be non-discrete. The utility of player $i$ in the continuous model is the following:
\[\util{i} = \sum_{r=1}^m\int_{\used{i}{r}}^{\used{i}{r} + \ract{i}{r}}\val{t}{r}dt,\]
where $\used{i}{r} = \sum_{i'=1}^{i-1}\ract{i'}{r}$ is the amount already invested in resource $r$ by earlier players.

In this setting, in order to prove a theorem analogous to Theorem~\ref{thm:greedy} in the discrete setting, we need an analogue to Lemma~\ref{thm:greedy4} that holds in the full-information continuous setting. We no longer have the tight connection between our setting and matching; nonetheless, the fact that the greedy strategy is a $4$-approximation to $OPT$ continues to hold.

\begin{lemma}\label{lem:greedycont}
  The greedy strategy for many-to-one online, continuous,   resource-weighted ``matching'', where players arrive online and have   tuples of allowable volumes of resources, has a competitive ratio of   $\frac{1}{4}$.
\end{lemma}
\begin{proof}[Proof Sketch]
The proof is identical to the proof of Lemma~\ref{thm:greedy4}, with   the exception that we no longer want matchings $\mu, \mu'$ but   rather correspondences between continuous regions of $v'_r$. See   Figure~\ref{fig:cont} for a visual proof sketch. 
\end{proof}

\begin{figure}[ht!]\label{fig:cont}
 \begin{tikzpicture}
    \tikzset{
        hatch distance/.store in=\hatchdistance,
        hatch distance=10pt,
        hatch thickness/.store in=\hatchthickness,
        hatch thickness=2pt
    }

    \makeatletter
    \pgfdeclarepatternformonly[\hatchdistance,\hatchthickness]{flexible hatch}
    {\pgfqpoint{0pt}{0pt}}
    {\pgfqpoint{\hatchdistance}{\hatchdistance}}
    {\pgfpoint{\hatchdistance-1pt}{\hatchdistance-1pt}}%
    {
        \pgfsetcolor{\tikz@pattern@color}
        \pgfsetlinewidth{\hatchthickness}
        \pgfpathmoveto{\pgfqpoint{0pt}{0pt}}
        \pgfpathlineto{\pgfqpoint{\hatchdistance}{\hatchdistance}}
        \pgfusepath{stroke}
    }

    \begin{axis}[
        xmin=0,xmax=4,
        xlabel={$x_r$},
        ymin=0,ymax=1,
        axis on top,
        legend style={legend cell align=right,legend plot pos=right}] 
    \addplot[color=black,domain=0.03:5,samples=100] {1/(6 * (x+.5))};

    \addlegendentry{$v'_r$}
    \addplot+[mark=none,
        domain= 0.01:2.5,
        samples=100,
        pattern=flexible hatch,
        hatch distance=5pt,
        hatch thickness=0.5pt,
        draw=blue,
        pattern color=cyan,
        area legend]{1/(6 * (x+.5))} \closedcycle;    
        \addlegendentry{Optimal players'regions}
    \addplot+[mark=none,
        domain=0.01:.5,
        samples=100,
        pattern=flexible hatch,
        area legend,
        pattern color=red]{1/(6 * (x+.5))} \closedcycle;
    \addplot+[mark=none,
        domain=.8:1.3,
        samples=100,
        pattern=flexible hatch,
        area legend,
        pattern color=red]{1/(6 * (x+.5))} \closedcycle;
        \addlegendentry{Greedy players' regions}
    \addplot+[mark=none,
        domain=1.8:2.1,
        samples=100,
        pattern=flexible hatch,
        area legend,
        pattern color=red]{1/(6 * (x+.5))} \closedcycle;

    \end{axis}
\end{tikzpicture}
\caption{Suppose the blue regions are those selected by the players
  who got those regions in OPT, and the red regions are those selected
  by some other player. Then, if some greedy player(s) have taken at
  least half of the value of the optimal regions for another player,
  at least that much utility has been gained by the greedy players. If
  not, half the value is still available for the player at hand.}
\end{figure}
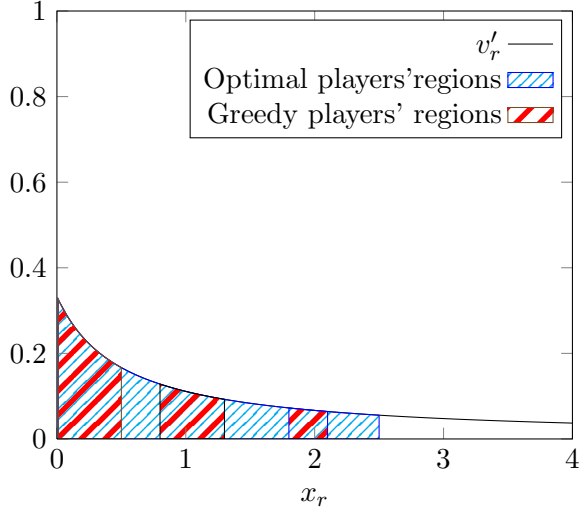

With Lemma~\ref{lem:greedycont}, following analysis similar to  Theorem~\ref{thm:greedy}, we have the following.

\begin{theorem}\label{thm:continuous}
  Suppose that \mech is an $(\mult, \add, \pfail)$-underestimating counter vector. Then, for any continuous,   future-independent resource-sharing game $g$,   $\comp_{\greedy}(\mech,g) = O(\mult \add)$.
\end{theorem}

\section{Resource Sharing with Future-Dependent utilities}
\label{section:future}
The second model of utility we consider is one where the benefit of
choosing a resource for a player depends not only the actions of the
past players but also on the actions taken by future
players. Specifically, all the players who selected a given resource
incur the same benefit regardless of the order in which they made the
choice. The utility of player $i$,
\[\util{i} = \dval{\w{r}}{r},\]
where $r$ is the resource chosen by player $i$ and $\w{r} =
\sum_{i'=1} ^n \ract{i'}{r}$ is the total utilization of resource $r$
by all players. As a warm-up, we start with the special case of market
sharing in the section below, and then move to the case of more
general value curves.

\subsection{Market sharing}
Market sharing is the special case of $\dval{x_{r}}{r} = c/x_{r}$ for
all $x_{r} \ge 1$. We show that with \emph{greedy} play and
\emph{private} counters, it is possible to achieve a logarithmic
factor approximation to the social welfare, while with
\emph{undominated} strategies and \emph{perfect} counters, one cannot
hope to achieve an approximation that is linear in the number of the
players.

\citet{Goemans:2004} showed that for market-sharing games, the
competitive ratio of $\mult$-approximate greedy play is at most
$O(\mult \log(n))$. Using analysis similar to theirs, we have the
following result.
\begin{corollary}\label{cor:marketlog}
  With $(\mult, \add, \pfail)$-counter vector and greedy play, with
  probability at least $1-\pfail$, the welfare achieved is at least
  $(OPT - 2\add \mult n)/O((1+\mult^{2}) \log(n))$.
\end{corollary}

For undominated strategies, we have the following result.
\begin{lemma}\label{lem:marketundom}
With perfect counters and undominated strategic play, there are games for which the welfare achieved is at most $OPT/(n\log(n))$. 
\end{lemma}
\begin{proof}
  Here is an example with $n$ players. Consider the case where for
  every $i \ge 1$, player $i$ is interested in resource 0 and resource
  $i$. For every $i\ge 1$, the total value of resource $i$ is
  $(n-i+1)(1-\epsilon)/i$ (for some small $\epsilon>0$). The value of
  resource 0 is 1.

  We claim that there is an undominated strategy game play where every
  player chooses resource 0 giving a social welfare of 1, whereas the
  optimal welfare is achieved by assigning player $i$ resource $i$
  giving a total welfare of $n(\log(n)-1)(1-\epsilon)$.

  Here is such an undominated strategy profile: for each $i$, player
  $i$ believes that every player after her is only interested in
  resource $i$. With this belief, it is easy to see that choosing
  resource 0 is an undominated strategy for every player.
\end{proof}

\subsection{General value curves}\label{section:future-appendix}
In this general setting, we will be interested in value curves that do
not decrease too quickly. Furthermore, we study only the greedy
strategy since we have already seen that undominated strategy does not
perform well even for simple curves (Lemma~\ref{lem:marketundom}).

\begin{definition}[$(w,l)$-shallow value curve]
  A value curve $\dvalo{r}$ is \shallow{\vol}{\length} if for all $x
  \leq \length$, it is the case that $\dval{x}{r} \ge
  \frac{\sum_{t=0}^{ x}\dval{t}{r}}{ \vol x}$.
\end{definition}
The definition of $(w,l)$-shallow value curve says that the actual
payoff all players get from the resource being utilized with $x$
weight is not too much smaller than the integral of $\dvalo{r}$ from
$0$ to $x$.

In the following result, we show that this restriction on the
\emph{rate of decay} of the value curves is necessary to say anything
nontrivial about the performance of the greedy strategy.
\begin{lemma}\label{lem:future-lb}
  Even with perfect counters, there exist sequential resource-sharing
  games $g$, where each resource $r$'s value curve $\dvalo{r}$ is
  $\shallow{w}{n}$, such that in the future-dependent setting,
  $\comp_\greedy(\full, g) \geq 2w$.
\end{lemma}
\begin{proof}
  Consider two players and two resources $r, r'$. Let $r$ have a value
  curve which is a step function, with $\val{0}{r} = w$, $\val{1}{r} =
  \frac{1}{2}$ and $\val{0}{r'} = w - \epsilon$. Suppose player one
  has access to both resources and player two has only resource $r$ as
  an option. Then, player one will choose $r$ according to greedy, and
  player two will always select $r$. The social welfare will be
  $SW(\greedy) = 1$, whereas $OPT$ is for player $1$ to take $r'$ and
  will have $SW(OPT) = 2w - \epsilon$. As $\epsilon \to 0$, this ratio
  approaches $2w$.
\end{proof}

Thus, as $w\to\infty$, the competitive ratio of the greedy strategy is
unbounded. Fortunately, the competitive ratio cannot be worse than
this, for fixed $w$, as we show in the theorem below.

\begin{theorem}\label{thm:full-dep}
  Suppose, for a sequential resource-sharing game $g$, each resource
  $r$'s value curve $v_r$ is $\shallow{w}{n}$. Then, in the in the
  future-dependent setting, $\comp_\greedy(\mech, g) =
  O(w\alpha\beta)$ for an $(\alpha, \beta)$-counter $\mech$.
\end{theorem}
\begin{proof}[Proof Sketch]
According to the greedy strategy, player $i$ chooses the resource in $A_{i}$ that maximizes $\dval{\used{i}{r}+1}{r}$, and we say $\dval{\used{i}{r}+1}{r}$ is her perceived value if $r$ is the resource she chose.
Terming the sum of the perceived values of all players as the perceived social welfare, $PSW(\greedy)$, we have
\begin{equation}\label{eqn:shallow}
\begin{split}
 PSW(\greedy) = & \sum_{i\in [n]} \sum_{r}\ract{i}{r} \dval{\used{i}{r}+1}{r} dx = \sum_{r}\sum_{x=0}^{\used{n}{r}+\ract{n}{r}}\dval{x}{r}\\
 \leq & \sum_{r} w \left(\used{n}{r} + \ract{n}{r}\right)\dval{\used{n}{r} + \ract{n}{r}}{r} = w\,SW(\greedy)
\end{split}
 \end{equation}
\noindent where the last inequality comes from our assumption about the value curves all being $\shallow{w}{n}$. 

The final part of the argument must show that the actual welfare from
greedy play with respect to the counters is well-approximated by the
perceived welfare with respect to the true counts. Since the counters
are accurate within some quantity $\leq n$, this is the case.
Following an analysis similar to that of Theorem \ref{thm:greedy},we
have our result.
\end{proof}

\section{Analysis of Private Counters}\label{section:app-counters}

\begin{proof}[Proof of Lemma~\ref{lemma:treesum}] We assume the reader is familiar with
  the TreeSum mechanism.
The privacy of this construction follows the same argument as for the
original constructions. One can view $m$ independent copies of the TreeSum protocol
as a single protocol where the Laplace mechanism is used to release
the entire vector of partial sums. Because the $\ell_1$-sensitivity of
each partial sum is 1 (since $\|a_t\|\leq 1$), the amount of Laplace
noise (per entry) needed to release the $m$-dimensional vector partial sums 
case is the same as for a dimensional $1$-dimensional counter. 

To see why the approximation claims holds, 
  we can apply Lemma 2.8 from \cite{ChanSS11} (a tail bound for sums
  of independent Laplace random variables) with $b_1=\cdots = b_{\log n} =
  {\log n}/\eps$, error probability $\delta = \pfail/mn$, 
  $\nu = \frac{ (\log n)\sqrt{\log(1/\delta)}}{\eps} $ and
  $\lambda = \frac{ (\log n)(\log (1/\delta)}{\eps}$, we get that each
  individual counter estimate $s_t(j)$ has additive error $O(\frac{(\log n )(\log
    (nm/\pfail))}{\eps})$ with probability at least $1-\pfail/(mn)$. Thus,
    all $n\cdot m$ estimates satisfy the bound simultaneously with probability at
    least $1-\pfail$.
\end{proof}
\begin{proof}[Proof of Lemma~\ref{lemma:FTSum}] 
  We begin with the proof of privacy. The first phase of the protocol
  is $\eps/2$-differentially private because it is an instance of the
  ``sparse vector'' technique of \citet{HR10} (see also \cite[Lecture
  20]{Roth-course11} for a self-contained exposition). The second
  phase of the protocol is $\eps/2$-differentially private by the
  privacy of TreeSum. Since differential privacy composes, the scheme
  as a whole is $\eps$-differentially private. Note that since we are
  proving $(\eps,0)$-differential privacy, it suffices to consider
  nonadaptive streams; the adaptive privacy definition then follows
  \cite{dwork2010}.
  
  We turn to proving the approximation guarantee. Note that the each
  of the Laplace noise variables added in phase 1 of the algorithm (to
  compute $\tilde{\used{t}{r}}$ and $\tau_j$) uses parameter
  $2/\eps'$. Taking a union bound over the $mn$ possible times that
  such noise is added, we see that with probability at least
  $1-\pfail/2$, each of these random variables has absolute value at
  most $O(\frac{\log(mn/\pfail)}{\eps'}$. Since $\frac {2}{\eps'}=
  O(\frac{mk}{\eps})$ and $k = O(\log \log (\frac {nm}{\pfail}) + \log
  \frac 1 \eps)$, we get that each of these noise variables has
  absolute value $\tilde O_\alpha (\frac{m \log(mn/\pfail)}{\eps})$
  with probability all but $\gamma/2$. We denote this bound $E_{1}$.

  Thus, for each counter, the $i$-th flag is raised no earlier than
  when the value of the counter first exceeds $\alpha^i (\log n) -
  E_1$, and no later than when the counter first exceeds $\alpha^i
  (\log n) + E_1$. The very first flag might be raised when counter
  has value 0. In that case, the additive error of the estimate is
  $\log n$, which is less than $E_1$. Hence, he mechanism's estimates
  during the first phase provide an
  $(\alpha,E_1,\pfail/2)$-approximation (as desired).
  
  The flag that causes the algorithm to enter the second phase is
  supposed to be raised when the counter takes the value
  $A:=\alpha^k(\log n) \geq \frac{\alpha}{\alpha -1}\cdot
  C_{tree}\cdot \frac{\log
      (nm/\gamma)}{\eps} $; in fact, the counter could be as small as
    $A-E_1$.  After that
  point, the additive error is due to the TreeSum protocol and is at most
  $B:= C_{tree}\cdot \log(n)\cdot \log(nm/\pfail)/\eps$ (with probability at least
  $1-\pfail/2$)
  by Lemma~\ref{lemma:treesum}. The reported value $\sused{i}{r}$ thus satisfies
  \[\sused{i}{r}\geq \used{i}{r} - B = \frac{1}{\alpha}\used{i}{r} +
  \underbrace{(1-\frac{1}{\alpha}) \used{i}{r} -B}_{\text{residual error}}\, .\]

  Since $\used{i}{r}\geq A-E_1$, the ``residual error'' in the
  equation above is at least $(1-\frac{1}{\alpha}) (A-E_1) - B =
  -(1-\frac{1}{\alpha})E_1\geq -E_1$. Thus, the second phase of the
  algorithm also provides $(\alpha, E_1,\pfail/2)$-approximation. With
  probability $1-\pfail$, both phases jointly provide a $(\alpha,
  E_1,\pfail)$-approximation, as desired.
\end{proof}


\end{document}